\documentclass[11pt]{article}
\usepackage{amsmath,amsfonts,latexsym,amssymb,graphicx,amsthm}
\usepackage{fullpage,color}
\usepackage{url,hyperref}
\usepackage{enumerate}
\usepackage[linesnumbered,boxed,ruled,vlined]{algorithm2e}

\newtheorem{lemma}{Lemma}
\newtheorem{theorem}[lemma]{Theorem}
\newtheorem{defn}[lemma]{Definition}

{\medskip}

\newcommand{\opt}{\Delta^*}
\newcommand{\tree}{\textsf{T}}
\newcommand{\parent}{\textsf{Par}}

\newlength{\boxlength}
\title{Purely Combinatorial Algorithms for Approximate Directed Minimum Degree Spanning Trees}
\author{Ran Duan \thanks{Institute for Interdisciplinary Information Sciences, Tsinghua University,  \href{}{duanran@mail.tsinghua.edu.cn}} 
	\and Tianyi Zhang \thanks{Institute for Interdisciplinary Information Sciences, Tsinghua University, \href{}{tianyi-z16@mails.tsinghua.edu.cn}}}
\date{}

\begin{document}

\maketitle

\thispagestyle{empty}
\begin{abstract}
Given a directed graph $G$ on $n$ vertices with a special vertex $s$, the directed minimum degree spanning tree problem requires computing a incoming spanning tree rooted at $s$ whose maximum tree in-degree is the smallest among all such trees. The problem is known to be NP-hard, since it generalizes the Hamiltonian path problem. The best LP-based polynomial time algorithm can achieve an approximation of $\opt+2$ [Bansal et al, 2009], where $\opt$ denotes the optimal maximum tree in-degree. As for purely combinatorial algorithms (algorithms that do not use LP), the best approximation is $O(\opt+\log n)$ [Krishnan and Raghavachari, 2001] but the running time is quasi-polynomial. In this paper, we focus on purely combinatorial algorithms and try to bridge the gap between LP-based approaches and purely combinatorial approaches. As a result, we propose a purely combinatorial polynomial time algorithm that also achieves an $O(\opt + \log n)$ approximation. Then we improve this algorithm to obtain a $(1+\epsilon)\opt + O(\frac{\log n}{\log\log n})$ for any constant $0<\epsilon<1$ approximation in polynomial time.
\end{abstract}


\section{Introduction}
Let $G = (V, E)$ be a directed graph, and $s\in V$ a special vertex designated as a sink. Conventionally, define $n = |V|$, $m = |E|$. A directed spanning tree of $G$ is a spanning tree rooted at $s$ such that every edge is directed from child to parent. The degree of a directed spanning tree is defined to be the maximum tree in-degree among all vertices. In the \emph{directed minimum degree spanning tree} problem (DMDST), we wish to compute a directed spanning tree of smallest degree. Denote the optimal degree by $\opt$. Computing the exact minimum is apparently NP-hard, because the Hamiltonian path problem can be reduced to DMDST. Therefore one should turn to look at the approximate version of DMDST.

There has been a line of research that focuses on this problem. In \cite{furer1990nc}, the authors proposed a polynomial time algorithm that finds a directed spanning tree whose degree is at most $O(\opt\log n)$. The approximation guarantee was improved to $b\opt + \log_b n, \forall b>1$ by \cite{dmdst}, but the time complexity was blown up to a quasi-polynomial time of $n^{O(\log_b n)}$. As shown in \cite{yao2008polynomial}, the approximate DMDST becomes much easier when $G$ is acyclic, where a directed spanning tree of degree $\leq \opt + 1$ can be computed in polynomial time. The currently best approximation was achieved by \cite{bansal2009additive} where an additive $2$-approximation, which is almost optimal, was guaranteed by an LP-based algorithm in polynomial time. 

\subsection{Our results}
In this paper we focus on approximation algorithms that are \emph{purely combinatorial} in the sense that we look for efficient algorithms that do not use linear programming. The reason why we study purely combinatorial algorithms for minimum-degree directed spanning trees is two-fold. On the one hand, purely combinatorial algorithms are interesting on their own right, because compared with LP-based algorithms, purely combinatorial algorithms are usually conceptually simpler and thus might be easier to implement in practice. On the other hand, the currently best purely combinatorial algorithm for the minimum-degree directed spanning tree problem runs in quasi-polynomial time and it only has an approximation guarantee of $O(\opt+\log n)$, and so there is still a large gap between LP-based approaches and purely combinatorial approaches in terms of both running time and approximation.

As a result, we show that the $O(\opt + \log n)$ approximation achieved by \cite{dmdst} can actually be achieved in polynomial time instead of quasi-polynomial time; furthermore, we extend the idea and improve the approximation to  $(1+O(\epsilon))\opt + O(\frac{\log n}{\log\log n})$.
\begin{theorem}\label{log}
There is an deterministic polynomial algorithm that computes a directed spanning tree of degree $O(\opt + \log n)$.
\end{theorem}

\begin{theorem}\label{loglog}
	For any fixed constant $\epsilon \in (0, \frac{1}{4})$, there is a deterministic polynomial time algorithm that computes a directed spanning tree of degree at most $(1+O(\epsilon))\opt + O(\frac{\log n}{\log\log n})$. Here $\log(\cdot)$ is calculated with base $2$.
\end{theorem}

To prove theorem \ref{log}, we start from the algorithmic framework of \cite{dmdst}. Their algorithm begins with an arbitrary directed spanning tree, and repeatedly conduct local searches. Every successful local search enables the algorithm to find a way to modify the current directed spanning tree such that a vertex with high tree in-degree loses a child at the cost that some vertices with low tree in-degree may gain a new child. When no such modifications can be found, an important lemma they have shown would argue a lower bound on $\opt$ which is related to some structures of the current directed spanning tree, thus ensuring an $O(\opt + \log n)$ approximation.

The time analysis of \cite{dmdst}'s algorithm is based on a potential function; more specifically, each vertex $u\in V$ is assigned a potential $\phi(u)$ which grows fast enough with $u$'s tree in-degree such that whenever a tree modification is carried out, the potential decrease due to the loss of a high-degree vertex always dominates the potential increase incurred by some low-degree vertices. The downside of their approach is that their algorithm does little effort to upper bound the potential increase from low-degree vertices during the local searches, and consequently it could come up with tree modifications where low-degree vertices contribute intensively to potential increase. In this case, the potential function drops fairly slowly and thus the time complexity becomes quasi-polynomial.

Our observation is that, instead of choosing a tree modification scheme instructed by an arbitrary successful local search, we could make a more prudent selection of successful local searches such that the potential increase brought about by low-degree vertices is small. As it turns out, this approach can significantly reduce the running time to a polynomial.

Now let us turn to Theorem~\ref{loglog}. Previously in Theorem~\ref{log}, when we have run out of good local searches, the algorithm would be stuck at an approximation of $O(\opt + \log n)$.  In order to go further than $O(\opt + \log n)$, we extend the original approach of local searches to what we will call ``augmenting paths''. Intuitively, although an unsuccessful local search does not directly give a way to lose one high-degree vertex, it might alter the structure of the current directed spanning tree such that new successful local searches may appear, and augmenting paths allow us to explore such possibilities. Speaking on a high-level, an augmenting path is a concatenation of a sequence of local searches which leads us, via a number of unsuccessful local searches,  finally to a successful one. In this way, augmenting paths may come up with more improvements on the spanning tree than the ordinary local searches are capable of.

\subsection{Related work}
There is a line of works that focus on minimum degree spanning trees in undirected graphs. For example there are two papers that focus on minimum degree spanning tree in undirected graphs. The first algorithm for approximate minimum degree spanning tree was proposed by \cite{furer1992approximating} where an $O(\opt + \log n)$ approximation can be computed in polynomial time. The approximation guarantee was shortly improved to the optimal $\opt + 1$ in \cite{furer1994approximating}, and the running time was improved to $\tilde{O}(mn)$.

There is also a line of works that are concerned with low-degree trees in weighted undirected graphs. In this scenario, the target low-degree that we wish to compute is constrained by two parameters: an upper bound $B$ on tree degree, an upper bound $C$ on the total weight summed over all tree edges. The problem was originally formulated in \cite{fischer1993optimizing}. Two subsequent papers \cite{konemann2000matter,konemann2003primal} proposed polynomial time algorithms that compute a tree with cost $\leq wC$ and degree $\leq \frac{w}{w-1}bB + \log_b n$, $\forall b, w> 1$. This result was substantially improved by \cite{chaudhuri2005would}; using certain augmenting path technique, their algorithm is capable of finding a tree with cost $\leq C$ and degree $B + O(\log n / \log\log n)$. Results and techniques from \cite{chaudhuri2005would} might sound similar to ours, but in directed graphs and undirected graphs we are actually faced with different technical difficulties. \cite{chaudhuri2005would}'s result was subsumed by \cite{goemans2006minimum} where $\forall k$, a spanning tree of degree $\leq k+2$ and of cost at most the cost of the optimum spanning tree of maximum degree at most $k$ can be computed in polynomial time. The degree bound was later further improved from $k+2$ to an optimal $k+1$ in \cite{singh2007approximating}.

Another variant is minimum degree Steiner trees which is related to network broadcasting \cite{ravi1994rapid,ravi1993many,fraigniaud2001approximation}. For undirected graphs, authors of \cite{furer1994approximating} showed that the same approximation guarantee and running time can be achieved as with minimum degree spanning trees in undirected graphs, i.e., a solution of tree degree $\opt + 1$ and a running time of $\tilde{O}(mn)$. For the directed case, \cite{fraigniaud2001approximation} showed that directed minimum degree Steiner trees problem cannot be approximated within $(1-\epsilon)\log|D|, \forall \epsilon > 0$ unless $\mathrm{NP}\subseteq \mathrm{DTIME}(n^{\log\log n})$, where $D$ is the set of terminals.

\section{Preliminaries}\label{notation}
Recalling the notations from the introductory section, $G = (V, E)$ is an arbitrary directed graph, and a sink $s\in V$. Let $n = |V|$, $m = |E|$. We assume $s$ is reachable from all vertices in $V$. Let $\opt$ be the tree degree of the optimal solution to the DMDST problem.

Both of algorithms in \ref{log} and \ref{loglog} will follow the algorithmic framework of \cite{dmdst}. During the course of our algorithms, we maintain a directed spanning tree $\tree$ and incrementally adjust the tree. For any $u\in V$, its tree in-degree, or simply degree, is the number of its children in $\tree$, which is denoted by $\deg(u)$, and let $\Delta = \max_{u\in V}\{\deg(u)\}$ keep track of the maximum tree in-degree of current $\tree$. For every $u\in V$ other than $s$, $\parent(u)$ refers to its parent in $\tree$. Let $\tree_u$ be the subtree of $\tree$ rooted at $u$. For each integer $d\in [0, \Delta]$, let $N_d$ be the set of all vertices whose tree in-degree is equal to $d$, and let $S_d$ be the set of vertices whose tree in-degree is $\geq d$. Simple path from $u$ to $v$ will be written as $u\rightsquigarrow v$. For the rest of this paper, $\log(\cdot)$ will have base $2$.

\begin{defn}
Two vertices $u, v\in V$ are {\rm unrelated} if $\tree_u\cap \tree_v$ is empty, that is, u is neither an ancestor nor a descendant of $v$; otherwise, $u, v$ are {\rm related}.
\end{defn}

Our algorithm are based on the iterative improvement procedure by improvement path, as in \cite{dmdst}. The improvement path is defined as:

\begin{defn}
Given a vertex $u\in V$ such that $\deg(\parent(u)) = d$, a simple path $u\rightsquigarrow w$ is called a {\rm $d$-improvement path}, if $w$ is the first vertex where this path steps out of $\tree_u$, and all vertices on this path other than $u$ have tree in-degree $\leq d-2$. A vertex $u$ is called {\rm $d$-improvable} if there exists a $d$-improvement that takes the form $u\rightsquigarrow w$.
\end{defn}

As discussed in \cite{dmdst}, $d$-improvement paths can be easily computed via a standard breath-first search which runs in polynomial time. It is also easy to see that if we have found a $d$-improvement path $u\rightsquigarrow w$, we can use the edges of this path as the outgoing edges in $T$ for all vertices on this path besides $w$, then the degrees of all vertices on $u\rightsquigarrow w$ besides $u$ will increase by 1, and the degree of $\parent(u)$ will decrease by 1. And all vertices in $T_u$ can reach $s$ in the updated tree, since the original paths to $s$ will be re-directed by  subpaths $u\rightsquigarrow w$. So we can get the following lemma, whose full proof is in~\cite{dmdst}.
\begin{lemma}[\cite{dmdst}]\label{improve}
Suppose $u\rightsquigarrow w$ is a $d$-improvement path. Then there is a way of adjusting $\tree$ in linear time so that: $\deg(\parent(u))$ decreases by 1; vertices on $u\rightsquigarrow w$ other than $u$ increase their tree in-degrees by at most 1; vertices not on $u\rightsquigarrow w$ do not change their tree in-degrees.
\end{lemma}

\section{Proof of Theorem~\ref{log}}\label{sect-log}

\subsection{Main algorithm}

Suppose in the current tree $T$ with maximum degree $\Delta$, if we can find a $\Delta$-improvement path, then by Lemma~\ref{improve} we can decrease the number of vertices with degree $\Delta$. However, sometimes such improvement paths is ``blocked'' by degree $\Delta-1$ vertices, so we need to first try to decrease the degree of those vertices. As we can see in Lemma~\ref{improve}, improving a degree $d$ vertices may increase the degrees of many other vertices, so it is more like a ``balancing'' procedure rather than simply ``improving''. Since it is hard to bound the number of degree increased vertices, \cite{dmdst} used a potential of $n^d$ for degree $d$ vertices to make sure the total potential decreases after each improvement step, thus had running time $n^{O(\log n)}$. By a careful selection of potential functions and improvement paths, we give a polynomial time algorithm for the same approximate ratio as~\cite{dmdst} in this section. In the next section, we further improve the approximate ratio while maintaining polynomial running time.

For each vertex $w\in V$, define its potential $\phi(w) = 2^{\deg(w)}$. Then define potential function to be the sum over all vertex potentials, i.e., $$\phi = \phi(\tree) = \sum_{w\in V}\phi(w) = \sum_{i = 0}^\Delta 2^i\cdot |N_i|$$

We propose the following algorithm \ref{poly}. Then we will prove that this algorithm computes an $O(\opt + \log n)$ approximation to the DMDST problem in polynomial time, which would immediately conclude theorem \ref{log}.

\begin{algorithm}
	\caption{Approximate DMDST algorithm}\label{poly}
	\While{$\Delta > 34\log  n$}{
		let $k = \arg\max_d \{2^d|N_d|\}$\;
		flag $=$ false\;
		\For{\rm $u$ whose parent is in $N_k$}{
			define $\psi_u = \sum_{v\in \tree_u\setminus S_{k-1}} 2^{\deg(v)}$\;
			\If{$\psi_u \leq 2^{k-3}$}{
				try to find a $k$-improvement path starting at $u$ using breath-first search\;
				if successful, flag $=$ true, and adjust $\tree$ accordingly\;
				\textbf{break}\;
			}
		}
		\If{\rm flag is false}{
			\Return $\tree$\;
		}
	}
	\Return $\tree$\;
\end{algorithm}

\subsection{Correctness}
We need two important lemmas from \cite{dmdst}.
\begin{lemma}[\cite{dmdst}]\label{witness}
Suppose there are subsets of vertices $U, B$ with the following two properties.
\begin{enumerate}[1.]\setlength{\itemsep}{0pt}
	\item Any path from $v\in U$ to $s$ must have an incoming edge into a vertex in $B$.
	\item For any two vertices $v,w \in U$, any path from $v$ to $s$ can intersect a path from $w$ to $s$ only after it passes through a vertex in $B$.
\end{enumerate}
In this case we call $B$ {\rm blocks} $U$. Then $\opt \geq |U| / |B|$.
\end{lemma}
\begin{proof}
	Let $\tree^*$ be an optimal directed spanning tree rooted at $s$. For all $u\neq U$, let $f_u\in B$ be the first vertex where the tree path from $u$ to $s$ intersects $B$; by property 1., $f_u$ exists for all $u\in U$. Let $\rho_u$ be the tree path from $u$ to $f_u$. By property 2., all $|U|$ paths $\{\rho_u\mid u\in U\}$ are disjoint. Since every $\rho_u$ has an incoming edge to some vertex in $B$, and by disjointness all these $|U|$ edges are different, by the pigeon-hole principle at least one of the vertex $v\in B$ has tree in-degree $\geq |U| / |B|$.
\end{proof}

\begin{lemma}[Implicit in \cite{dmdst}]\label{unrelated}
For any $d$, there are at least $(d-1)|N_d| + 1$ unrelated vertices whose parents are in $N_d$.
\end{lemma}
\begin{proof}
	We basically follow the same lines as in \cite{dmdst}. The proof is by induction on the cardinality of $N_d$.
	\begin{itemize}\setlength{\itemsep}{0pt}
		\item Basis: $|N_d| = 1$. The single vertex of this set has exactly $d$ children which are unrelated.
		\item Induction: $|N_d| > 1$. Find a vertex $u\in N_d$ such that no other $v\in N_d\setminus \{u\}$ is a descendent of $u$ (for example, the $u$ with the largest depth in $\tree$). By induction, there are at least $(d-1)(|N_d|-1) + 1$ vertices whose parents are $\in N_d\setminus \{u\}$. Since these vertices are unrelated, at most one of them, say $w$, is an ancestor of $u$. Then, removing $w$ (if it exists) and adding all $d$ children of $u$, we have obtained a set of $\geq (d-1)(|N_d|-1) + 1 -1 + d = (d-1)|N_d| + 1$ unrelated vertices whose parents are in $N_d$.
	\end{itemize}
\end{proof}

We prove when the algorithm terminates, $\Delta = O(\opt + \log n)$. If the algorithm terminates on line-12, then $\Delta \leq 34 \log n = O(\opt + \log n)$. Now we assume the algorithm terminates on line-11.

Firstly we observe that $k\geq \Delta - \log n > 33\log n$; this is because for any $d<\Delta - \log n$, $2^d |N_d| < 2^{\Delta - \log n} n \leq 2^{\Delta} \leq 2^\Delta |N_\Delta|$. By lemma \ref{unrelated}, we can find a subset of unrelated vertices $W$ whose parents are in $N_k$ and $|W| > (k-1)|N_k|$. When the algorithm terminates on line-11, we can only guarantee that there is no $k$-improvement path from $u\in W$ with smaller $\psi(u)$. However, we can bound the number of $u\in W$ that still has a $k$-improvement path from $u$, so that:
\begin{lemma}\label{not-improvable}
When the algorithm terminates, the number of unrelated children of $N_k$ that are not $k$-improvable path is at least $(k-1)|N_k|/2$.
\end{lemma}
\begin{proof}
By lemma~\ref{unrelated}, find a subset of unrelated vertices $W$ whose parents are in $N_k$ and $|W| > (k-1)|N_k|$.
Since all vertices in $W$ are unrelated, we have:
$$\sum_{u\in W, k\text{-improvable}} \psi_u = \sum_{u\in W, k\text{-improvable}} \sum_{v\in \tree_u\setminus S_{k-1}} 2^{\deg(v)}\leq \sum_{d=0}^{k-2}2^d |N_d|$$
Suppose at least half of $u\in W$ are $k$-improvable, by the pigeon hole principle, there exists $u\in W$ with a $k$-improvement path, such that
$$\begin{aligned}
\psi_u &\leq \frac{2}{|W|}\sum_{d=0}^{k-2} 2^d |N_d| < \frac{2}{|W|}\sum_{d = \Delta - 2\log n}^{k-2}2^d |N_d| + \frac{2}{|W|} n 2^{\Delta - 2\log n -1}\\
&\leq \frac{2}{|W|}\cdot 2\log n\cdot 2^{k}|N_k| + \frac{2}{|W|} n 2^{\Delta - 2\log n -1}	&	(\text{maximality of }k)\\
&\leq \frac{2}{|W|}\cdot 2\log n\cdot 2^{k}|N_k| + \frac{2}{|W|}n 2^{k-\log n -1}	&	(k\geq \Delta - \log n)\\
&< \frac{2}{(k-1)|N_k|} \cdot 2\log n \cdot 2^{k}|N_k| + \frac{2}{(k-1)|N_k|}  2^{k-1}	&	(|W| > (k-1)|N_k|)\\
&\leq \frac{32}{33} 2^{k-3} + \frac{8}{33\log n} 2^{k-3} \leq 2^{k-3}	&	(n\geq 256)
\end{aligned}$$
This contradicts the termination condition of our algorithm because the flag variable would be set true. Hence, at least half of $u\in W$ are not $k$-improvable, and we get the conclusion.
\end{proof}

Let $U$ be the set of all such vertices from Lemma~\ref{not-improvable}, then we claim $S_{k-1}$ blocks $U$. Check the two properties described in lemma \ref{witness}.
\begin{enumerate}[(1)]\setlength{\itemsep}{0pt}
	\item Let $u\rightsquigarrow s$ be any path from $u\in U$ to the sink $s$. Let $w$ be the first vertex on $u\rightsquigarrow s$ out of subtree $\tree_u$. As $u$ is not $k$-improvable, at least one vertex on $u\rightsquigarrow w$ other than $u$ belongs to $S_{k-1}$.
	\item Let $u, v\in U$ be two different vertices, and $u\rightsquigarrow s$, $v\rightsquigarrow s$ are paths to the sink $s$ from $u, v$ respectively. As $U$ only consists of unrelated vertices, before the two paths $u\rightsquigarrow s$ and $v\rightsquigarrow s$ intersect, one of them, say $u$, must have stepped out of the subtree $\tree_u$, then by the same arguments in (1), $u\rightsquigarrow s$ has intersected $S_{k-1}$ before it meets $v\rightsquigarrow s$.
\end{enumerate}

Then by Lemma~\ref{witness} and Lemma~\ref{not-improvable}, we have $$\opt \geq |U| / |S_{k-1}| \geq \frac{(k-1)|N_k|}{2|S_{k-1}|}$$

By maximality of $2^k|N_k|$, $|N_k|\geq 2^i |N_{k+i}|$ for all $i$. Therefore, $$|S_{k-1}|= \sum_{i=-1}^{\Delta -k}|N_{k+i}| \leq \sum_{i=-1}^{\Delta -k}2^{-i} |N_{k}|\leq 4|N_k|$$

Hence, $$\opt \geq \frac{(k-1)|N_k|}{2|S_{k-1}|}\geq \frac{(k-1)|N_k|}{8|N_k|}\geq \frac{k-1}{8} \geq \frac{\Delta - \log n-1}{8}$$ which immediately yields $\Delta = O(\opt + \log n)$.

\subsection{Running time}
We upper bound the total running time by an analysis of the potential $\phi$. On the one hand, consider those vertices whose potential decrease after a tree adjustment associated with a $k$-improvement. By lemma \ref{improve}, one vertex from $N_k$ encounters a decrement of in-degree which induces a decrease of $2^{k-1}$ of vertex potential. On the other hand, consider those vertices whose vertex potentials increase after such an adjustment. Potential increase accumulated from $\tree_u$ would not exceed $\psi_u\leq 2^{k-3}$ because, as stated in lemma \ref{improve}, all vertices in $\tree_u$ increase their tree in-degrees by at most 1; plus, the vertex potential increase outside $\tree_u$ would not exceed $2^{k-2}$. Therefore, the overall potential loss is at least $2^{k-1} - 2^{k-2} - 2^{k-3} = 2^{k-3}$. As $k\geq \Delta - \log n$, 
$$2^{k-3}\geq \frac{1}{8n} 2^{\Delta} = \frac{1}{8n^2} n 2^\Delta = \Omega(\frac{\phi}{n^2})$$
The last inequality holds because the largest in-degree is currently $\Delta$ and consequently $\phi = \sum_{v\in V}2^{\deg(v)} \leq n\cdot 2^\Delta$.

In other words, each improvement path reduces $\phi$ by a factor of $1 - \Omega(1/n^2)$. Since $n < \phi < n\cdot 2^n$, the total number of iterations is bounded by $$\log_{1 - \Omega(1/n^2)}2^{-n}=O(n^3)$$. Each iteration is dominated by $O(n)$ breath-first searches, which then takes $O(mn)$ time. Hence the total running time would be $O(mn^4)$ which is polynomial.
\section{Proof of Theorem~\ref{loglog}}\label{sect-loglog}
\subsection{Sketch of algorithm}
In this section we will prove theorem \ref{loglog}. The first ingredient is that we extend the idea of $k$-improvement path. Intuitively speaking, we relax the condition that a $k$-improvement path must end at a vertex whose tree in-degree is strictly less than $k-1$ by allowing endpoints to have tree in-degree of exactly $k-1$. If this endpoint is of $k-1$ tree in-degree, then in order to eventually reduce the size of $|N_k|$, we need to find a path that starts at one of its children, traversing the corresponding subtree, and end at another vertex of tree in-degree $\leq k-1$. We repeat this procedure until we end at a vertex of tree in-degree strictly less than $k-1$. In a nutshell, we compute a sequence of paths, which will be called an \emph{augmenting path}, so that, after we adjust the $\tree$ for each path, $|N_k|$ can be decreased.

There are two technical hindrances to this approach. 
\begin{enumerate}[(1)]\setlength{\itemsep}{0pt}
	\item The sequence of paths may intersect with themselves. If this could happen, our adjustment may increase some values of $\deg(\cdot)$ significantly. To resolve this issue, an important observation is that if we restrict ourselves to subtrees that contain no vertices of tree in-degree $\geq k-2$, then this sequence of paths basically would not have intermediate self-intersections, and it still works if it ends with some vertex that already appeared.
	\item As we adjust along many paths, a huge number of vertices may increase tree in-degrees, which leads to a significant increase of potential. If we do not have good upper bounds on such kind of potential increase, it may neutralize the potential decrease brought about by the decrement of $|N_k|$.
	
	To overcome this difficulty, first note that potential increase resulted by a path is more or less bounded by the total vertex potential below the subtree rooted at the start-vertex of this path. So to give bounds on potential increase, it suffices to bound the sum of all vertex potentials in all of these subtrees, say $\tree_{u_1}, \tree_{u_2}, \cdots, \tree_{u_l}$. The core of our algorithm is that, we try to make sure $\sum_{w\in \tree_{u_i} \setminus S_{k-2}}\phi(w)$ decreases geometrically with respect to $i$, say $\sum_{w\in \tree_{u_i} \setminus S_{k-2}}\phi(w)\propto 1 / (1+\epsilon)^i$. If at one point, after we already have found $u_1, u_2, \cdots, u_{h-1}$, we could no longer find a $u_h$ that satisfies the desired geometric bound, we then apply lemma \ref{witness} to prove a lower bound on $\opt$ and terminate the algorithm.
\end{enumerate}

\subsection{Augmenting paths}
Formulate the idea of augmenting paths in the definition below.
\begin{defn}
A {\rm $k$-augmenting path} is a sequence of $l$ simple paths $u_1\rightsquigarrow v_1, u_2\rightsquigarrow v_2, \cdots\cdots, u_l\rightsquigarrow v_l$ with properties below.
\begin{enumerate}[(i)]\setlength{\itemsep}{0pt}
\item $v_i = \parent(u_{i+1}), \forall i<l$.
\item All $u_i$'s are unrelated, and all $v_i$'s are different.
\item $\deg(\parent(u_1)) = k$; $\deg(v_i) = k-1$, $\forall 1\leq i < l$; $\deg(v_l)\leq k-1$.
\item Every subtree $\tree_{u_i}$ does not contain any vertices of tree in-degree $\geq k-2$.
\item For any path $u_i\rightsquigarrow v_i, i\in[l]$, it does not contain any vertices of tree in-degree $\geq k-2$ except for $v_i$, and $v_i$ is the first vertex where this path steps out of the subtree $\tree_{u_i}$.
\end{enumerate}
\end{defn}

We now argue in the next lemma that any $k$-augmenting path gives a tree adjustment scheme for decreasing the size of $N_k$.
\begin{lemma}\label{aug}
Let $u_1\rightsquigarrow v_1, u_2\rightsquigarrow v_2, \cdots\cdots, u_l\rightsquigarrow v_l$ be a $k$-augmenting path. If $\deg(v_l)\leq k-2$, then we can adjust tree $\tree$ in using this $k$-augmenting path in polynomial time so that the following requirements are met.
\begin{enumerate}[(1)]\setlength{\itemsep}{0pt}
	\item $\parent(u_1)$ moves from $N_k$ to $N_{k-1}$; namely its tree in-degree decreases by 1.
	\item If $\deg(v_l) \leq k-3$ before the adjustment is carried out, then $\deg(v_l)$ increases by at most 2 afterwards; otherwise $\deg(v_l) = k-1$. All other $\deg(v_i)$ stay unchanged.
	\item All vertices other than $u_i, v_i, i\in [l]$ on this $k$-augmenting path increases their tree in-degree by at most 1.
	\item All $|N_i|$'s do not increase, $\forall i>k$.
\end{enumerate}
\end{lemma}
\begin{proof}
	Consider this procedure: for each simple path $u_i\rightsquigarrow v_i$, denote it by $u_i = w_1\rightarrow w_2\rightarrow\cdots\rightarrow w_h = v_i$. For every $1\leq j<h$, cut $w_j$ off of its parent in $\tree$ and append it right below $w_{j+1}$. Then only vertices on this $k$-augmenting path may increase in-degrees, so we only focus on these vertices.
	
	We claim all paths $u_i\rightsquigarrow v_i$'s do not intersect except at vertex $v_l$. In fact, for any two different $u_i\rightsquigarrow v_i$ and $u_j\rightsquigarrow v_j$ that intersect with each other, as $u_i$ and $u_j$ are unrelated, the intersection point can only be $v_i$ or $v_j$ due to (\romannumeral5); then, again by (\romannumeral2) $v_i\neq v_j$, one of them belongs to the opposite subtree, say $v_j\in \tree_{u_i}$. Since according to (\romannumeral4) $\tree_{u_i}$ does not contain any vertices of tree in-degree $\geq k-2$, it must be $\deg(v_j) < k-2$, and thus $j = l$. Therefore, before the last vertex $v_l$ accepts a new branch, 
	\begin{enumerate}[(a)]\setlength{\itemsep}{0pt}
		\item $\parent(u_1)$ becomes of $(k-1)$-degree, and all $v_i$ stays $(k-1)$-degree.
		\item any vertex on $u_i\rightsquigarrow v_i$ encounters an increase of tree in-degree by at most 1;
	\end{enumerate}
	As paths $u_i\rightsquigarrow v_i$ contain no vertices of tree in-degree $\geq k-2$ except for $v_i$, none of those vertices belong to any $N_j, j>k$ after the adjustment. Hence no $|N_j|, j>k$ has increased.
	
	Now consider the last step when $v_l$ accepts a new branch. There are three cases.
	
	\begin{enumerate}[(a)]\setlength{\itemsep}{0pt}
		\item If $v_l$ has appeared twice in this $k$-augmenting path, then $v_l$ is contained in an $\tree_{u_i}$ at the beginning, and thus its in-degree was $\leq k-3$ right from the start, as $\tree_{u_i}\cap (N_{k-1}\cup N_{k-2})$ was empty. As $v_l$'s in-degree increases by at most 2, its tree in-degree is strictly less than $k$ after all of our cut-and-append procedures.
		\item If $v_l$ has appeared only once in this $k$-augmenting path, then its tree in-degree increases by at most 1, and thus its tree in-degree is $\leq k-2 + 1 = k-1$.
		\item $v_l$ could never appear in the $k$-augmenting path for more than twice because all subtrees $\tree_{u_i}, i\in[l]$ are disjoint.
	\end{enumerate}
	
	Therefore, $\deg(v_l)$ stays smaller than $k$ after the adjustment. To summarise, $|N_k|$ decreases by 1, and no other $|N_i|, i>k$ has increased.
\end{proof}

\subsection{Main algorithm}
Assume $\epsilon \in (0, \frac{1}{4})$ is a constant. Define $c = 2\log^{0.4}n$. We assume $c\geq 4$, $c > 1/\epsilon$; this assumption is valid when $n$ is sufficiently large. For each vertex $w\in V$, define its potential $\phi(w) = c^{\deg(w)}$. Then define potential function to be the sum over all vertex potentials, i.e., $$\phi = \phi(\tree) = \sum_{w\in V}\phi(w) = \sum_{i=0}^\Delta c^i\cdot |N_i|$$

One problem with $k$-augmenting path is that when we adjust the tree $\tree$ as in lemma \ref{aug}, we could blow up the potential function $\phi$ significantly. Therefore, we should only focus on $k$-augmenting with some additional nice properties.

\begin{defn}
	A $k$-augmenting path specified by $u_1\rightsquigarrow v_1, u_2\rightsquigarrow v_2, \cdots\cdots, u_l\rightsquigarrow v_l$ is called {\rm potential-efficient} if the inequality holds: (remind that $\tree_{u_i}\cap S_{k-2}=\emptyset$ for every subtree $\tree_{u_i}$)
	$$\sum_{w\in \tree_{u_i}} c^{\deg(w)}\leq 0.9\cdot \frac{\epsilon}{(1+\epsilon)^i}\cdot c^{k-1}$$
\end{defn}

The lemma below lower-bounds the potential decrease when lemma \ref{aug} is applied on a potential-efficient $k$-augmenting path.
\begin{lemma}\label{adjust}
	When lemma \ref{aug} is applied on a potential-efficient $k$-augmenting path, the potential decrease is at least $0.05\cdot c^k$.
\end{lemma}
\begin{proof}
	Consider firstly those vertices whose potential have increased.	Since all vertices on $u_i\rightsquigarrow v_i$, which belong to $\tree_{u_i}$ except for $v_i$, increase their tree in-degrees by at most 1, the total potential increase caused by these vertices is at most
	$$(c-1)\cdot\sum_{v\in \tree_{u_i}} c^{\deg(v)}\leq (c-1)\cdot 0.9\cdot \frac{\epsilon}{(1+\epsilon)^i}\cdot c^{k-1}$$
	
	For all $v_i, i\in [l]$, $\phi(v_i)$ does not change expect $v_l$. By lemma \ref{adjust}, the potential increase of $\phi(v_l)$ is at most $c^{k-1} - c^{k-2}$, for sufficiently large $n$. Summing up, the total potential increase would be
	$$\leq c^{k-1} - c^{k-2} + \sum_{i=1}^l(c-1)\cdot 0.9\cdot \frac{\epsilon}{(1+\epsilon)^i}\cdot c^{k-1} < c^{k-1} - c^{k-2} + 0.9 c^k < 0.94 c^k$$
	for sufficiently large $n$.
	
	Now secondly let us consider those vertices whose potential have decreased. By lemma \ref{adjust}, as a $k$-degree vertex, $\parent(u_1)$ has lost a child which is $u_1$ from such a tree adjustment. Hence the decrease of $\phi(\parent(u_1))$ is equal to $c^{k} - c^{k-1}$.
	
	Overall, there the potential loss is $> c^{k} - c^{k-1} - 0.94\cdot c^{k} > 0.05 c^{k}$ when $n$ is large enough.
\end{proof}

Similar to the original algorithm in \cite{dmdst}, our algorithm operates iteratively, and each iteration consists of two major steps.
\begin{enumerate}[Step (1)]\setlength{\itemsep}{0pt}
	\item Find a proper $k\leq \Delta$ as well as a set $W$ of unrelated vertices as the starting vertices of potential-efficient $k$-augmenting paths.
	
	\item We search for longer and longer potential-efficient $k$-augmenting paths starting from vertices in $W$. Eventually if we successfully find a potential-efficient augmenting path, then we adjust tree $\tree$ accordingly which would greatly decrease the potential function $\phi$, and then move on to the next iteration; otherwise, we argue a lower bound on $\opt$ of $(1 - O(\epsilon))\cdot(\Delta - O(\frac{\log n}{\log\log n}))$ and terminate the algorithm.
\end{enumerate}
The whole procedure is described in Algorithm~\ref{improved}.

\begin{algorithm}
	\caption{Improved approximate DMDST algorithm}\label{improved}
	\While{$\Delta > 2\frac{\log n}{\log (c/2)}$}{
		let $k = \arg\max_d \{\left(\frac{c}{2}\right)^d \cdot |N_d|\}$\;
		flag $=$ false\;
		$V_0=N_k$\;
		\Repeat{$|\bigcup_{j=0}^i V_j|<(1+\epsilon)|\bigcup_{j=0}^{i-1} V_j|$} {
			$i=i+1$, $V_i=\emptyset$ \;
			Define $U_i$ to be the set of vertices $u$ whose parent is in $V_{i-1}$ satisfying $\tree_u\cap S_{k-2}=\emptyset$ and $\sum_{w\in \tree_u} c^{\deg(w)}\leq 0.9\cdot \frac{\epsilon}{(1+\epsilon)^i}\cdot c^{k-1}$ \;
			\For{\rm $u\in U_i$} {
				find the set of vertices $X$ outside $\tree_u$ that $u$ can reach by a path whose intermediate vertices are in $\tree_u$ \;
				\If {$X$ contains a vertex of degree $\leq k-2$} {
					flag $=$ true, and adjust $\tree$ according to Lemma~\ref{adjust}\;
					\textbf{break}\;
					}
				$V_i=V_i\cup (X\cap N_{k-1})$ \;
			}
			$V_i = V_i \setminus \bigcup_{j=0}^{i-1}V_j$\;
		}
		\If{\rm flag is false}{
			\Return $\tree$\;
		}
	}
	\Return $\tree$\;
\end{algorithm}

\subsubsection*{Analysis}

\begin{lemma}\label{bound-k}
In step 2, the chosen $k$ is at least $\Delta - \frac{\log n}{\log (c/2)}$.
\end{lemma}
\begin{proof}
For $d< \Delta - \frac{\log n}{\log (c/2)}$,
$$(c/2)^d|N_d| \leq n\cdot (c/2)^d = n\cdot (c/2)^{d - \Delta}\cdot (c/2)^\Delta < n\cdot (c/2)^{-\frac{\log n}{\log (c/2)}}\cdot (c/2)^\Delta = (c/2)^\Delta \leq (c/2)^\Delta |N_\Delta|$$
So those $d$ cannot be chosen.
\end{proof}

If algorithm \ref{improved} terminates on line-18, then $\Delta \leq 2\frac{\log n}{\log(c/2)} = \frac{5\log n}{\log\log n}$ by the time of termination, which already gives an approximation of $(1 + O(\epsilon))\cdot \opt + O(\frac{\log n}{\log\log n})$. For the rest of this section, we only consider terminations on line-17.

In the algorithm, the degree of vertices in $V_0$ is $k$, and the degree of vertices in $V_1,V_2,\cdots$ is $k-1$. 
We will ensure that the sets $\{V_i\}$ are disjoint, and then, since the subtrees rooted at $u\in U_i$ does not contain vertices of degree $\geq k-2$, all vertices in $U_i$ are unrelated and all sets $\{U_i\}$ are disjoint.
So when an iteration ends with a true ``flag'', we can find a $k$-augmenting path $u_1\rightsquigarrow v_1, u_2\rightsquigarrow v_2, \cdots\cdots, u_l\rightsquigarrow v_l$ with $u_i\in U_i$ and $v_i\in V_i$,
and by Lemma~\ref{adjust}, the potential decrease is at least 
$$0.05\cdot c^k\geq 0.05\cdot c^{\Delta-\frac{\log n}{\log (c/2)}}\geq 0.05\cdot\frac{1}{n}\phi(T)\cdot 2^{-\log n}\cdot (c/2)^{-\frac{\log n}{\log (c/2)}}=\frac{1}{20n^3}\phi(T)$$
Thus, we can conclude that 
\begin{lemma}
The total number of big iterations in Algorithm~\ref{improved} is bounded by a polynomial of $n$.
\end{lemma}
\begin{proof}
The big iteration can only continue to the next one when we have found a $k$-augmenting path, and we have shown the potential will be decreased by a factor of $\leq (1-\frac{1}{20n^3})$.
Since $n\cdot c\leq \phi(T)\leq c^n$, the total number of this improvement is bounded by 
$$\frac{\log c^{n-1}}{-\log (1-\frac{1}{20n^3})}=O(n^4\log c)$$
\end{proof}

Next, we analyze the case when the algorithm ends with $|\bigcup_{j=0}^i V_j|<(1+\epsilon)|\bigcup_{j=0}^{i-1} V_j|$. In every iteration, the set $U_i$ chosen in step 7 which are the set of vertices $u$ satisfying:
\begin{enumerate}[(a)]
	\item The parent of $u$ is in $V_{i-1}$, where $V_0=N_k$.
	\item Subtree $T_u$ contains no vertices of degree $\geq k-2$.
	\item $\sum_{w\in \tree_u} c^{\deg(w)}\leq 0.9\cdot \frac{\epsilon}{(1+\epsilon)^i}\cdot c^{k-1}$.
\end{enumerate}

First, if the inner loop does not end, we have $|\bigcup_{j=0}^i V_j|\geq (1+\epsilon)|\bigcup_{j=0}^{i-1} V_j|$, which means $|\bigcup_{j=0}^i V_j| \geq (1+\epsilon)^i |N_k|$,
also $|V_i|\geq \epsilon|\bigcup_{j=0}^{i-1} V_j|$, so we have $|V_i|\geq \epsilon  (1+\epsilon)^{i-1} |N_k|$. 

We can lower bound the total size of $U_1,U_2,\cdots,U_i$ by:
\begin{lemma}\label{bound-Ui}
The size of $U_i$ after iteration $i$ is at least $(k-2-\frac{c^2}{\epsilon})|V_{i-1}|$.
\end{lemma}
\begin{proof}
The degree of vertices in $V_{i-1}$ is $k$ or $k-1$, then by lemma \ref{unrelated}, we can find a set $W$ of unrelated vertices satisfying (a), such that $|W|\geq (k-2)|V_{i-1}| + 1$. 
Then the next thing we do is to remove from $W$ all vertices that fail to meet requirements (b)(c) to obtain a lower bound of the size of all $U_i$, since vertices in $U_i$ are children of $V_{i-1}$.
Thus, we need to upper-bound the total number of vertices in $W$ that violate either (b) or (c).
	\begin{itemize}\setlength{\itemsep}{0pt}
		\item Upper bound on violations of (b).
		
		By maximality of $(c/2)^k|N_k|$, for all degree $d$, we have $|N_d|\leq (c/2)^{k-d}|N_k|$. So 
		$$\sum_{d=k-2}^{\Delta} |N_d| \leq |N_k|\cdot \sum_{i=-2}^{\infty}(\frac{c}{2})^{-i} =\frac{c^3}{4(c-2)}|N_k|$$
		Since $c\geq 4$, $c/(c-2)\leq 2$, the number of subtrees violates (b) is bounded by $\frac{c^2}{2}|N_k|$.
		
		\item Upper bound on violations of (c).
		First, we can bound the total potential of vertices of degree less than $k-2$. By maximality of $(c/2)^k|N_k|$, $c^d|N_d|\leq 2^{d-k}c^k|N_k|$. Then
		$$\sum_{d=0}^{k-3} c^d|N_d| \leq c^k|N_k|\cdot \sum_{d=0}^{k-3} 2^{d-k}<\frac{c^k}{4}|N_k|$$
		Note that the potentials of the subtrees violating (c) must be larger than $0.9\cdot \frac{\epsilon}{(1+\epsilon)^i}\cdot c^{k-1}$. Because vertices in $W$ are unrelated,
		the number of subtrees rooted at $W$ violating (c) is bounded by
		$$\frac{\frac{c^k}{4}\cdot|N_k|}{0.9\cdot \frac{\epsilon}{(1+\epsilon)^i}\cdot c^{k-1}} \leq \frac{(1+\epsilon)^i c}{3.6\epsilon} |N_k| \leq \frac{(1+\epsilon)^i c^2}{3.6}|N_k|$$
		where the last inequality follows from $\frac{1}{c} < \epsilon <\frac{1}{4}$.
	\end{itemize}

	Combining both aspects, as $|V_{i-1}|\geq \epsilon(1+\epsilon)^{i-2}|N_k|$, we immediately derive a lower bound on the size of $U_i$, i.e. 
	$$\begin{aligned}
	&|U_i| \geq |W| - (\frac{c^2}{2} + \frac{(1+\epsilon)^i c^2}{3.6})|N_k|\\
	&> (k-2)|V_{i-1}| - (\frac{c^2}{2} + \frac{(1+\epsilon)^i c^2}{3.6})\cdot \frac{1}{\epsilon (1+\epsilon)^{i-2}}|V_{i-1}|\\
	&> (k-2)|V_{i-1}| - (\frac{c^2}{2\epsilon(1+\epsilon)^{i-1}} + \frac{(1+\epsilon)^2c^2}{3.6\epsilon})|V_{i-1}|\\
	&> (k-2)|V_{i-1}| - \frac{c^2}{\epsilon}|V_{i-1}| = (k-2 - \frac{c^2}{\epsilon})|V_{i-1}|
	\end{aligned}$$
	which finishes the proof.

\end{proof}

Since the size of union of $\{V_i\}$ will increase by a factor of at least $(1+\epsilon)$ in each iteration, the total number of iterations of inner loop is bounded by $O(\epsilon^{-1}\log n)$.
After reaching $|\bigcup_{j=0}^i V_j|< (1+\epsilon)|\bigcup_{j=0}^{i-1} V_j|$, from the same argument of Algorithm~\ref{poly}, we can check:
\begin{itemize}
	\item Since every $X$ in step 10 does not contain vertices of degree $\leq k-2$, the first vertex $u\in U_i$ can reach outside $\tree_u$ is of degree $\geq k-1$, which are included in $S_{k+1}$ and 
	$\bigcup_{j=0}^i V_j$.
	\item Let $u,v\in \bigcup_{j=1}^i U_j$, since $u,v$ are unrelated, the path from $u$ and $v$ to $s$ cannot intersect before they go out of $T_u$ and $T_v$, that is, before they reach $S_{k+1}$ or 
	$\bigcup_{j=0}^i V_j$.
\end{itemize}
So by lemma \ref{witness}, from $V_0=N_k$, the optimal degree $\opt$ is bounded by:
$$\begin{aligned}
&\opt \geq \frac{|\bigcup_{j=1}^i U_j|}{|\bigcup_{j=0}^i V_j|+|S_{k+1}|} = \frac{\sum_{j=1}^i |U_j|}{\sum_{j=0}^i |V_j|+|S_{k+1}|}	&	(\text{disjointness of }U_j, V_j)\\
&\geq \frac{(k-2 - \frac{c^2}{\epsilon})\sum_{j=0}^{i-1}|V_j|}{\sum_{j=0}^i |V_j|+|S_{k+1}|}	&	(\text{lemma \ref{bound-Ui}})
\end{aligned}$$

We can bound $|S_{k+1}|=\sum_{d=k+1}^{\Delta} |N_d| \leq |N_k|\cdot \sum_{i=1}^{\infty}(\frac{c}{2})^{-i} =\frac{2}{c-2}|N_k|$. Since $1/c<\epsilon<1/4$, 
$|S_{k+1}|< \frac{2\epsilon}{1-2\epsilon} |N_k| <4\epsilon |\bigcup_{j=0}^{i-1} V_j|$.
Since $k>\Delta - \frac{\log n}{\log (c/2)}> \frac{\log n}{\log (c/2)}$, we assume $k>2c^2/\epsilon^2>2/\epsilon + c^2/\epsilon^2$, thus

$$\opt \geq \frac{(k-2-\frac{c^2}{\epsilon})|\sum_{j=0}^{i-1} |V_j|}{(1+\epsilon)\sum_{j=0}^{i-1} |V_j|+4\epsilon \sum_{j=0}^{i-1} |V_j|} > \frac{1-\epsilon}{1+5\epsilon}k>(1-6\epsilon)(\Delta - \frac{\log n}{\log (c/2)})$$

Finally, as $\epsilon$ is a fixed constant, and $c=2\log^{0.4} n$, then by lemma \ref{bound-k}, the assumption $k > \Delta - \frac{\log n}{\log(c/2)} > \frac{\log n}{\log(c/2)} > 2c^2/\epsilon^2$ holds when $n$ is sufficiently large. So finally we can have a directed spanning tree with in-degree $\Delta>(1+O(\epsilon))\opt + O(\frac{\log n}{\log\log n})$.

\vspace{5mm}
\bibliographystyle{plain}
\bibliography{dmdst}

\end{document}